\documentclass[letterpaper, 10 pt, conference]{ieeeconf}  

\IEEEoverridecommandlockouts                              

\usepackage{amsmath} 
\usepackage{amssymb}  
\usepackage{kotex}
\usepackage{color}
\usepackage{cite}
\usepackage{graphicx} 
\usepackage{tikz} 
\usepackage{pgfmath} 
\usepackage{balance} 
\usepackage{subfigure}
\usepackage{multirow}
\usepackage{bm}

\newtheorem{assumption}{Assumption}
\newtheorem{definition}{Definition}
\newtheorem{remark}{Remark}
\newtheorem{theorem}{Theorem}
\newtheorem{lemma}{Lemma}
\newtheorem{corollary}{Corollary}
\allowdisplaybreaks

\title{\LARGE \bf 
A Nonlinear Scaling-based Design of Control Lyapunov-barrier Function for Relative Degree 2 Case and its Application to Safe Feedback Linearization
}

\author{Haechan Pyon and Gyunghoon Park$^{*}$
\thanks{H. Pyon and G. Park are with School of Electrical and Computer Engineering, University of Seoul, Republic of Korea. {\{\tt\small hc.pyon, gyunghoon.park@uos.ac.kr\}}}
\thanks{
*Corresponding author: Gyunghoon Park
}%
\thanks{
{This work was supported by the Ministry of Trade, Industry, \& Energy (MOTIE, Korea) under the Industrial Technology Innovation Program, Grant No. 20023606, ‘‘Rapid reconfigurable robotic workcell technology for fast automation and modification in cell production processes including assembly, packaging, and inspection’’.}
}
}

\begin{document}
\allowdisplaybreaks

\maketitle
\thispagestyle{empty}
\pagestyle{empty}

\begin{abstract}
In this paper we address the problem of control Lyapunov-barrier function (CLBF)-based safe stabilization for a class of nonlinear control-affine systems.
A difficulty may arise for the case when a constraint has the relative degree larger than $1$, at which computing a proper CLBF is not straightforward. 
Instead of adding an (possibly non-existent) control barrier function (CBF) to a control Lyapunov function (CLF), our key idea is to simply scale the value of the CLF on the unsafe set, by utilizing a sigmoid function as a scaling factor.
We provide a systematic design method for the CLBF, with a detailed condition for the parameters of the sigmoid function to satisfy. 
It is also seen that the proposed approach to the CLBF design can be applied to the problem of task-space control for a planar robot manipulator with guaranteed safety, for which a safe feedback linearization-based controller is presented. 

\end{abstract}

\section{Introduction}


As modern control systems are predicted to operate more often in the workspace shared with humans and/or other objects, {\it safety} has gained a significant attention in control society. 
The quantity of how a system is safe has been measured by a so-called {\it barrier function} (BF), which has been used to characterize an unsafe state by having a specific sign or sufficiently large value, so that the safety can be analyzed by observing a change of the value of the BF. 
The control barrier function (CBF) is an extended version of the BF defined for the systems having an input, based on which many researchers have proposed various safety-critical control schemes \cite{wieland2007constructive, ames2019control}. 


As safety and stability are regarded as two important control objectives to be satisfied, it is not surprising that a variety of research efforts have been paid to resolving {\it the problem of safe stabilization}.
A considerable approach to controller design for the safe stabilization is based on the {\it control Lyapunov-barrier function} (CLBF) that has mixed properties of the control Lyapunov function (CLF) for stability and the CBF for safety  \cite{romdlony2016stabilization, romdlony2014uniting, wu2019control}. 
A CLBF has been designed to remain positive inside the unsafe set (as similar to the conventional CBFs), and to admit a control input with which the CLBF decreases as time goes on (like other CLFs), for whose design the {\it Sontag's universal formula} can be applied \cite{lin1991universal}. 

Despite its success in the literature, however, the CLBF has been employed mostly for limited cases, as it is usually chosen as a linear combination of a CLF and a (conventional) CBF so that safety-related constraints cannot be handled properly. 
As a solution to overcoming this drawback, in this paper we propose a new design methodology for the CLBF, with which a class of constraints having relative degree $2$ can be considered.  
For this, we employ a nonlinear scaling to a CLF (instead of linearly combining the CBF with the CLF), by weighting the existing CLF on the unsafe set with sigmoid function as scaling factors.
In doing so, the proposed CLBF does not explicitly utilize a CBF in its design, which allows to deal with higher relative degree cases. 
It should be pointed out that for ease of construction, we restrict ourselves to the class of the 2nd-order nonlinear affine systems, while extension to more general case is left to future works.   
We theoretically derive a sufficient condition for the parameters of the sigmoid function (i.e., scaling function) with which the resultant scaled CLF becomes a CLBF. 
Moreover, reporting a new finding that the CLBF-based controller can be designed as an {\it add-on type} controller attached to an existing baseline controller, we also show below that the proposed method for the CLBF design can be successfully applied to the problem of {\it safe} task-space control for planar robot manipulators.

Before closing this section, it is remarked that for ensuring safety of a system having high relative degree, several attempts have been made in the literature. 
Most of the relevant works have focused on extending the notion of the CBF to the higher-order case, which inevitably requires a set of inequalities to be solvable at the same time.
Unfortunately, guaranteeing the simultaneous feasibility of these inequalities is not trivial, and even the structure of the resultant CBF-based controller becomes relatively complicated as higher relative degree is considered.  
Compared with the existing high-order CBF approaches, the proposed one has two distinguishable advantages: first, we theoretically show that it is always possible to construct a proper CLBF for the case of relative degree $2$; moreover, once a CLBF is well-selected, the Sontag's universal formula-based control law can still be utilized for the high relative degree case, which leads to simplicity in controller design. 

{\it Notations.} 
For a scalar-valued function $V$ and the vector-valued function $f$, the notation $L_fV(x)$ is used for the Lie derivative $({\partial V}/{\partial x})f(x)$. The space $C^1(\mathbb{R}^n, \mathbb{R})$ consists of all continuously differentiable functions $F:\mathbb{R}^n \to \mathbb{R}$. 
For a square and symmetric matrix $P$, $\lambda_{m}(P)$ and $\lambda_{M}(P)$ denote the smallest and largest eigenvalue of $P$, respectively.
The matrix $I_n \in \mathbb{R}^{n \times n}$ stands for the $n \times n$ identity matrix. 
Given a set $\mathcal{D}$, denote the boundary of $\mathcal{D}$ by $\partial \mathcal{D}$, the closure of $\mathcal{D}$ by ${\rm cl}(\mathcal{D})$, and the interior of $\mathcal{D}$ by $\mathrm{int}(\mathcal{D})$, respectively.

\section{Elementary Result: Nonlinear Scaling-based CLBF Design for Relative Degree 2 Cases}\label{sec:basic}

\subsection{Problem Formulation}
Consider a 2nd-order, nonlinear, single-input control-affine system 
\begin{equation} \label{eq:system}
    \begin{aligned}
        \dot{x}_1 &= x_2,\\
        \dot{x}_2 &= f(x) + g(x)u
    \end{aligned}
\end{equation}
where $x=(x_1,x_2)\in \mathbb{R}^2$ is the state vector, $u\in \mathbb{R}$ is the input, and the functions $f$ and $g$ are assumed to be Lipschitz on the region of interest $\mathcal{X}$ that is compact and contains the origin $x=0$.
Moreover, $g(x)\neq 0$ for all $x\in \mathcal{X}$. 
For ease of explanation, define $F(x):=\begin{bmatrix}
    x_2 & f(x)
\end{bmatrix}^\top$ and $G(x):= \begin{bmatrix}
    0 & g(x)
\end{bmatrix}^\top$ so that \eqref{eq:system} can be compactly rewritten as $\dot{x} = F(x) + G(x)u$. 

A formal definition of the safety of a system is given as in the following. 
\begin{definition}{(Safety, \cite{wieland2007constructive})}\label{def:safety} 
 Given a set of unsafe states $\mathcal{D}$, an autonomous system $\dot{x} = F(x)$ is said to be {\it safe} if the state $x(t)$ satisfies
 \begin{equation} \label{eq:unsafecondition}
 x(t) \notin \text{cl}(\mathcal{D}), \quad \forall t \geq 0.
 \end{equation}
 $\hfill\square$
\end{definition} 

When considering a non-autonomous system \eqref{eq:system}, it is usually desired to construct a state-feedback controller $u=\kappa(x)$
with which stability and safety are achieved simultaneously, in a sense that 
\begin{itemize}
    \item the origin $x=0$ of the closed-loop system $\dot{x} = F(x) + G(x)\kappa(x)$ is asymptotically stable, and 
    \item the closed-loop system $\dot{x} = F(x) + G(x)\kappa(x)$ is safe (so that $x(t)$ satisfies \eqref{eq:unsafecondition}). 
\end{itemize}
Throughout this paper, such a problem is called the {\it safe stabilization} problem.

As a tool for safe stabilization of a non-autonomous system, the control Lyapunov-barrier function (CLBF) has been studied for a decade in the relevant literature, which is the main topic of this work. 
An usual definition of the CLBF is stated below. 
\begin{definition}{(CLBF, \cite{romdlony2016stabilization})}
    Given a system $\dot{x}=F(x) + G(x)u$ and an unsafe set $\mathcal{D} \subset \mathbb{R}^n$, a proper and lower-bounded function $W \in C^1 (\mathbb{R}^n, \mathbb{R})$ satisfying
\begin{subequations} \label{eq:CLBF}
    \begin{align}
        & W(x) >0 \quad \forall  x \in \mathcal{D} \label{eq:CLBF_a} \\
        & L_FW(x) < 0\quad  \forall x \in \{z \in \mathbb{R}^n \setminus (\mathcal{D} \cup \{0\}) : L_GW(z) = 0\} \label{eq:CLBF_b} \\
        & \mathcal{U} := \{x \in \mathbb{R}^n : W(x) \leq 0\} \neq \emptyset \label{eq:CLBF_c} \\
        & \text{cl}(\mathbb{R}^n \setminus (\mathcal{D} \cup \mathcal{U})) \cap \text{cl}(\mathcal{D}) = \emptyset \label{eq:CLBF_d}
    \end{align}
\end{subequations}
is called a {\it control Lyapunov-barrier function}. 
$\hfill\square$
\end{definition}

The main aim of this section is to construct a CLBF $W(x)$ for the non-autonomous system \eqref{eq:system} and for given unsafe set $\mathcal{D}$. 
Before going on further, it is noted that the condition~\eqref{eq:CLBF_d} for $W$ being a CLBF is needed particularly when global safe stabilization is desired, as remarked in \cite{romdlony2016stabilization}. 
In our case, since the region $\mathcal{X}$ of our interest is supposed to be bounded, there seems to be no reason to think about \eqref{eq:CLBF_d} anymore. 
With this kept in mind, the following definition introduces a {\it weaker} version of the conventional CLBF, in which the domain of functions is restricted to the bounded set $\mathcal{X}$.  
\begin{definition}{(Weak CLBF)}\label{def:weak_CLBF}
    Given a system $\dot{x}=F(x) + G(x)u$, a compact set $\mathcal{X}\subset \mathbb{R}^n$, and an unsafe set $\mathcal{D} \subset \mathcal{X}$, a proper and lower-bounded function $W \in C^1 (\mathbb{R}^n, \mathbb{R})$ is called a {\it weak CLBF} with respect to $\mathcal{D}$ if
\begin{subequations} \label{eq:weak_CLBF}
    \begin{align}
        & W(x) >0 \quad \forall  x \in \mathcal{D} \label{eq:weak_CLBF_a} \\
        & L_FW(x) < 0\quad  \forall x \in \{z \in \mathcal{X} \setminus (\mathcal{D} \cup \{0\}) : L_GW(z) = 0\} \label{eq:weak_CLBF_b} \\
        & \mathcal{U} := \{x \in \mathcal{X} : W(x) \leq 0\} \neq \emptyset. \label{eq:weak_CLBF_c}    
        \end{align}
\end{subequations}
 $\hfill\square$
\end{definition}

In the rest of this section, we propose a design method for a weak CLBF, provided that the unsafe set $\mathcal{D}$ to be considered is defined with a linear inequality but has relative degree\footnote{The \textit {relative degree} of a function $B$ with respect to $\dot{x} = F(x) + G(x)u$ is defined as the number of times we need to differentiate the function along the system dynamics until  $u$ explicitly shows.} larger than $1$, as follows:
\begin{assumption}\label{asm:unsafe_set}
    For the system \eqref{eq:system}, the unsafe set $\mathcal{D} \subset \mathbb{R}^{2}$ is given by
    \begin{align}
        \mathcal{D} = \{ (x_1,x_2)\in \mathcal{X}: x_1\leq d\} \label{ineq:unsafeset}
    \end{align}
    where $ d<0$. $\hfill\square$
\end{assumption}

It should be noted that, even though the defining condition for $\mathcal{D}$ seems to be simple enough, the problem of constructing a CLBF $W(x)$ for $\mathcal{D}$ in \eqref{ineq:unsafeset} remains challenging.
This is mainly because, the defining inequality $x_1\leq d$ for $\mathcal{D}$ is independent of $x_2$. 
    Indeed, most CLBFs presented in the literature have the form of $ W = V + \theta B -k$ with a 
    CLF $V$ and a 
    CBF $B$ with respect to $\mathcal{D}$.
    On the other hand, as a barrier candidate for $\mathcal{D}$ above, $B(x)$ is expected to be a function of only $x_1$. 
    This directly implies that $B(x)$ has the relative degree larger than $1$, or equivalently, $\dot{B}(x) = (\partial B/\partial x_1)\dot{x}_1 = (\partial B/\partial x_1)x_2$ does not depend on $u$.
    Thus $\dot{B}(x) = L_F B(x)$ may fail to be negative at some $x$ satisfying $L_G W(x) = 0$, so that there is no guarantee for $B(x)$ to be a CBF. 

We close this subsection by introducing an extra assumption that the system \eqref{eq:system} itself is asymptotically stable and admits a quadratic-type control Lyapunov function (CLF).



\begin{assumption}\label{asm:CLF}
    For the system \eqref{eq:system}, there exists a square matrix $P^{\top}=P >0$ such that the corresponding Lyapunov function candidate
    \begin{subequations}\label{asm:stability}
    \begin{align}\label{asm:stability_a}
        &  V(x) = \frac{1}{2}x^{\top}Px
    \end{align}
    satisfies
    \begin{align}\label{asm:stability_b}
        & \frac{\partial V}{\partial x} F(x) <0, \quad \forall x \in \{z \in \mathbb{R}^2\setminus\{0\} : L_GV(z) =0 \}.
    \end{align}
    \end{subequations} $\hfill\square$
\end{assumption}

It should be noted in advance that, while we assume the stability of \eqref{eq:system} as above, this requirement can be readily relaxed when the theory is applied for the controller design. 
We will revisit this point shortly in the upcoming section.



\begin{remark}
    Even though there seem to be two possible cases of defining an unsafe set defined with $x_1$ (one for $x_1 \leq d$, while the other for $x_1 \geq d$), we simply assume the first one because one in fact implies the other. 
    Indeed, even if our system is constrained as in the second case, we can reach the opposite case by introducing a new coordinate $\hat x = (\hat x_1, \hat x_2) := (-x_1, -x_2)$ and redefining the system dynamics $\hat f (\hat x) := f(-x)$, $\hat g (\hat x) := g(-x)$,  without loss of generality. 
    $\hfill\square$
\end{remark}

\subsection{Nonlinear Scaling-based Design of Weak CLBF}

In this section, we propose a new method to construct a weak CLBF for the 2nd-order nonlinear system \eqref{eq:system} with the unsafe set $\mathcal{D}$ in \eqref{ineq:unsafeset}.
Instead of relying on the existence of the CBF as in the existing work \cite{romdlony2016stabilization} (which may not be trivial for cases of higher relative degree), our core idea is to simply rescale the existing CLF $V(x)$ in \eqref{asm:stability_a} in a sense that, the resulting function has a relatively larger value near the unsafe set.
More specifically, a candidate for the weak CLBF to be proposed is given by 
\begin{align} \label{CLBF_nonlinear_scaled}
    W(x) := \big(1 + \theta \sigma(x_1)\big) V(x) - k
\end{align}
where $V(x)$ is the quadratic CLF in \eqref{asm:stability_a}, and $\theta>0$ and $k >0$ are parameters to be determined later. 
Most importantly, the function $\sigma(x_1)$ of $x_1$ represents the quantity of rescaling the CLF $V(x)$ with $\mathcal{D}$ taken into account. 
In this work, we construct the rescaling function $\sigma$ as a sigmoid function
\begin{align}\label{sigmoid}
    \sigma (x_1) = \frac{1}{1 + \exp(l(x_1 - d - (\delta/2)))}
\end{align}
where $d$ is the constant that defines $\mathcal{D}$, $l>0$ determines the slope of the sigmoid function $\sigma$, and $\delta>0$ can be understood as a margin for safety. 
Since the scaling factor $\theta\sigma$ for $W(x)$ is a nonlinear function of $x_1$, we term the proposed approach to constructing a weak CLBF a {\it nonlinear scaling approach} throughout this work.

We now have four parameters $l$, $\theta$, $k$, and $\delta$ left to be determined, whose detailed conditions for selection will be provided in a sequel.  
Before going on further, recall first that the CLF condition for $V(x)$ guarantees existence of control input that ensure the forward invariance of a level set
 \begin{align}\label{Omega1_set}
     \Omega := \{x\in \mathcal{X} : V(x) \leq v_2\}
 \end{align}
 which can be understood as a (sub)set of possible initial conditions $x(0)$.
It is obvious that, if the system is initiated in a small $\Omega$ such that $\Omega \cap {\rm cl}(\mathcal{D}) = \emptyset$, the system is inherently safe (i.e., $x(t)\notin {\rm cl}(\mathcal{D})$ for all $t$).
With this in mind, to allow a non-trivial set of initial values, in what follows we consider the level set $\Omega$ \eqref{Omega1_set} with sufficiently large $v_2$, without loss of generality, that satisfies 
\begin{align}\label{avoidingtrivalcase}
    \Omega \cap \mathcal{D} \neq \emptyset
\end{align}
which allows to avoid some trivial cases. 
We will show below that a subset 
\begin{align}\label{C_Omega}
    \mathcal{C}_{\Omega} := \left\{ x \in \mathcal{X} :
V(x) \leq v_2 , \:
x_1 \geq d + \delta
\right\}
\end{align}
of $\Omega \cap \big(\mathcal{X}\setminus \mathcal{D}\big)$ 
is indeed a set of the initial conditions $x(0)$ with which the associated state $x(t)$ is able to remain the safe set $\mathcal{X}\setminus \mathcal{D}$ for the entire time period with a proper control.

\begin{figure}
\centering
    \subfigure[Level set $\Omega$]{\includegraphics[width = 0.32\linewidth]{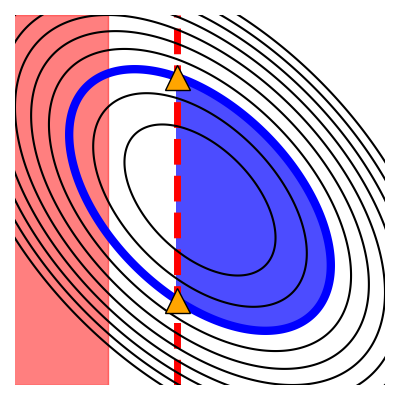}}
    \hfill
    \subfigure[Safe set $\mathcal{U}$]{\includegraphics[width=0.32\linewidth]{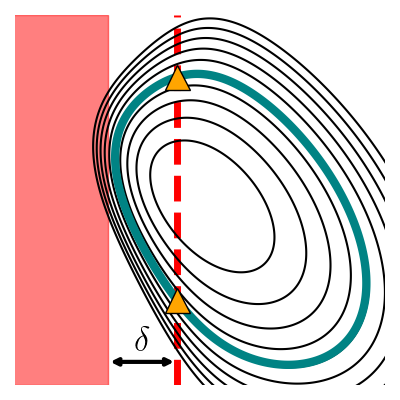}}
    \hfill
    \subfigure[Both sets]{\includegraphics[width = 0.32\linewidth]{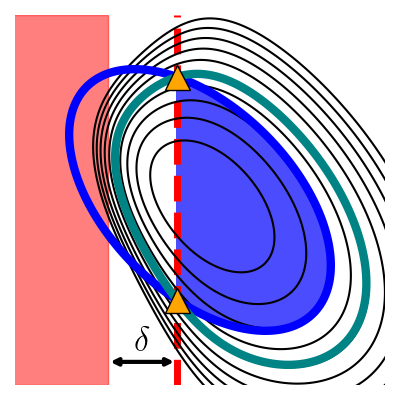}
    }
    \caption{Illustration of the unsafe set $\mathcal{D}$ (red), the level set $\Omega$ with the level $v_2$ of the CLF $V$, and the safe set $\mathcal{U}$ of the proposed CLBF $W$ with the given level $v_2$. Marked points (orange triangle) are intersections of $\partial \Omega$ and area satisfying $x_1 = d + \delta$ (red dashed line), where a rescaled value $(1+\theta \sigma(d + \delta))v_2$ of the level set $\Omega$ at this point is selected as DC offset $k$ of $W$. The set $\mathcal{C}_{\Omega}$ (shown as blue coloured area in $(\textrm{a})$) remains in $\mathcal{U}$ of $W$.}
    \label{set comparison}
\end{figure}



With the level $v_2$ of $\Omega$ given, the parameters can be taken via the following steps. 
First, the slope $l$ of the sigmoid function $\sigma$ is chosen such that
\begin{align} \label{l_choosing}
    0<l \leq \frac{2}{\gamma},\quad \gamma:=\sup_{x\in \mathcal{X}} \: x_1 < \infty. 
\end{align}
In turn, take the safety margin $\delta$ in \eqref{sigmoid} as 
\begin{align}\label{delta_choosing}
    \delta > \frac{2}{l}\cdot\ln\left(\frac{v_2}{v_1}\right) \quad v_1 := \min_{x \in \mathcal{D}}V(x)
\end{align}
where $v_1$ must be smaller than $v_2$ due to \eqref{avoidingtrivalcase}. 
Next, choose the scaling factor $\theta$ in \eqref{CLBF_nonlinear_scaled} such that
\begin{align} \label{theta_choosing}
\theta > \frac{v_2 - v_1}{\sigma_1v_1 - \sigma_2v_2}
\end{align}
with $\sigma_1 := \sigma(d)$ and $\sigma_2 := \sigma(d+\delta)$. 
It is important to note that $\sigma_1v_1 - \sigma_2v_2>0$ holds with $\delta$ satisfying \eqref{delta_choosing}, as
\begin{align*}
    \frac{\sigma_1}{\sigma_2} = \frac{1+e^{({l\delta}/{2})}}{1+e^{-({l\delta}/{2})}} = e^{({l\delta}/{2})} >\frac{v_2}{v_1}.
\end{align*}
At last, we select $k$ for the DC offset of \eqref{CLBF_nonlinear_scaled} as
\begin{align} \label{k_choosing}
    k := (1+\theta \sigma_2)v_2.
\end{align}





We are now ready to present our main theoretical result, which says that the nonlinearly-scaled  $W(x)$ in \eqref{CLBF_nonlinear_scaled} is indeed a weak CLBF for the system \eqref{eq:system} and the unsafe set \eqref{ineq:unsafeset}.

\begin{theorem}\label{Main_theorem}
Suppose that Assumptions~\ref{asm:unsafe_set} and \ref{asm:CLF} hold, and the parameters $l$, $\delta$, $\theta$, $k$ are chosen such that \eqref{l_choosing}, \eqref{delta_choosing}, \eqref{theta_choosing}, and \eqref{k_choosing} are satisfied. 
   Then the function $W(x)$ of $x$ in  \eqref{CLBF_nonlinear_scaled} is a weak CLBF with respect to $\mathcal{D}$. 
   Moreover, there is no other stationary point in $\mathcal{U}=\{x \in \mathcal{X} : W(x) \leq 0\}$ except the origin. $\hfill\square$
 \end{theorem}

\begin{proof}
Note first that the function $W$ is obviously proper and lower-bounded in $\mathbb{R}^2$ by
\begin{align*}
    \lambda_m(P)\|x\|^2 -k \leq W(x) \leq (1+\theta)\lambda_M(P)\|x\|^2 -k
\end{align*}
and $-k \leq W(x)$, $\forall x \in \mathbb{R}^2$. 
In the rest of the proof, we will show that all the conditions in \eqref{eq:weak_CLBF} hold, from which $W(x)$ is proven to be a weak CLBF.

{\it (Condition~\eqref{eq:weak_CLBF_a})} It is noted that $\sigma <\sigma_1$ in $\mathcal{D}$ by the construction of  $\sigma$ in \eqref{sigmoid}. 
Then $\forall x \in \mathcal{D}$, we have  
{
\begin{align*}
W(x) & = (1 + \theta \sigma)V - k\\
& \geq (1 + \theta \sigma_1)v_1 - (1+\theta \sigma_2)v_2\\
& = (v_1 - v_2) + \theta (\sigma_1 v_1 - \sigma_2 v_2)\\
& = (\sigma_1 v_1 - \sigma_2 v_2)
\left(
\theta - \frac{v_2-v_1}{\sigma_1 v_1 - \sigma_2 v_2}
\right) > 0
\end{align*} 
} 
where the last inequality is carried out by utilizing \eqref{delta_choosing} and \eqref{theta_choosing}, which concludes the proof of the first item. 

{\it (Condition~\eqref{eq:weak_CLBF_b})} For the proof of \eqref{eq:weak_CLBF_b}, we first consider $L_GW(x)$ that is expressed as
\begin{align*}
L_GW(x) = \frac{\partial W}{\partial x_2}g(x)=(1+\theta \sigma)\frac{\partial V}{\partial x_2}g(x)
\end{align*}
which vanishes if and only if 
\begin{align}\label{L_GW=0_condition}
    \frac{\partial V}{\partial x_2} = p_{12}x_1 + p_{22}x_2 =0
\end{align}
where $p_{ij}$ represents the $(i,j)$-component of the matrix $P$. 
It is noted for a future use that every $p_{ij}$ are positive, while the proof of this claim is given in Appendix. 
On the subspace defined by \eqref{L_GW=0_condition}, the quadratic CLF $V(x)$ has the value
\begin{align*}
V(x) = V^\circ(x_1) := \frac{{\rm det}(P)}{2p_{22}}
x_1^2
\end{align*}
by simply replacing $x_2$ in $V(x)$ with $-(p_{12}x_1)/(p_{22})$.
On the other hand, one has
\begin{align*}
\frac{\partial W}{\partial x_1} 
& = \theta V \frac{\partial \sigma}{\partial x_1}  + (1 + \theta \sigma) \frac{\partial V}{\partial x_1}\\
& = \theta V (-l\sigma(1-\sigma)) + (1 + \theta \sigma)(p_{11}x_1 + p_{12}x_2)
\end{align*}
where we have used an useful property of the sigmoid function $\sigma$: that is, $(\partial \sigma)/(\partial x_1) = -l\sigma (1-\sigma)$. 
Thus, on the subspace at which \eqref{L_GW=0_condition} holds, we have
\begin{align}
\frac{\partial W}{\partial x_1}  & = -V^\circ \theta l\sigma(1-\sigma) + (1 + \theta \sigma) 
\left(
p_{11}x_1 - \frac{(p_{12})^2}{p_{22}}x_1
\right)\notag \\
&=  - \dfrac{{\rm det}(P)}{2p_{22}} \theta l\sigma(1-\sigma)x_1^2
 + \frac{(1 + \theta \sigma)}{p_{22}}({\rm det}(P))x_1 \notag\\
& = 
\frac{\theta \sigma {\rm det}(P)}{p_{22}}
\left(-\frac{l}{2}(1-\sigma)x_1 + 1 \right)x_1+\frac{{\rm det}(P)}{p_{22}}x_1 \notag\\
& = \frac{{\rm det}(P)}{p_{22}}
\bigg[\theta \sigma \left(-\frac{l}{2}(1-\sigma)x_1 + 1 \right)+1\bigg]
x_1.\label{eq:dWdx1}
\end{align}
Finally, the term $L_FW(x)$ is computed as
\begin{align*}
& L_FW(x) = \begin{bmatrix}
\dfrac{\partial W}{\partial x_1} & 
\dfrac{\partial W}{\partial x_2}\end{bmatrix} 
\begin{bmatrix}
x_2\\
f
\end{bmatrix} = \frac{\partial W}{\partial x_1} x_2
\\
& = \dfrac{{\rm det}(P)}{p_{22}}
\bigg[
\theta \sigma 
\left(
1 -\dfrac{l}{2}(1-\sigma)x_1 
\right)
+
1
\bigg]
x_1
\left(
-\dfrac{p_{12}}{p_{22}}x_1
\right)
\\
& = -\dfrac{p_{12}{\rm det}(P)}{(p_{22})^2}x_1^2
\:
\cdot
\: 
\bigg[
\theta \sigma 
\left(
1 -\dfrac{l}{2}(1-\sigma)x_1 
\right)
+
1
\bigg].
\end{align*}
Since $p_{12}>0$ (see Appendix) and $0<\sigma<1$ by taking $\sigma$ as the sigmoid function \eqref{sigmoid}, 
\begin{align}\label{condition_l_inproof}
    1-\frac{l}{2}(1-\sigma)x_1 > 0
\end{align}
implies $L_FW(x) <0$ in the region under consideration.
The condition
\eqref{condition_l_inproof} in fact holds, because
\begin{align*}
1-\frac{l}{2}(1-\sigma)x_1 
& \geq  1-\frac{l}{2}(1-\sigma)\gamma\\
& > 1 - \frac{l}{2}\gamma\\
& = \frac{\gamma}{2}
\left(
\frac{2}{\gamma} - l
\right) \geq 0
\end{align*}
by selecting $l$ to satisfy \eqref{l_choosing}.
This completes the proof of the second item. 

{\it (Condition~\eqref{eq:weak_CLBF_c})} The third item is obviously true, because $W(0) = -k <0$. 

{\it (No stationary point)} From now on, we will prove that there is no stationary point on $\mathcal{U}$, except the origin.  
At the stationary point $x^*$, it is necessary that
\begin{align*}
        \frac{\partial W}{\partial x_1}(x^*) = \frac{\partial W}{\partial x_2}(x^*) =0.
\end{align*}
One can readily obtain (by following the same procedure for deriving \eqref{condition_l_inproof}) that, provided that $\partial W/\partial x_2 =  g(x)(p_{12}x_1 + p_{22}x_2)=0$ at $x=x^*$,
$\partial W/\partial x_1$ must have the form \eqref{eq:dWdx1}.
Since the parameters are selected to satisfy \eqref{condition_l_inproof}, $\partial W/\partial x_1$ vanishes only if $x_1 = 0$.
On the other hand, we already have $ p_{12}x_1 + p_{22}x_2 = 0$, which means that only $x^*=0$ makes $\partial W/\partial x$ be zero. 
It means that, the stationary point in $\mathcal{U}$ is uniquely determined as $x^* =0$.
\end{proof}



As a byproduct, the following lemma shows that $\mathcal{C}_\Omega$ in \eqref{C_Omega} can be a candidate for the set of initial conditions $x(0)$ with which the associated state trajectory $x(t)$ remains the safe set. 

\begin{lemma} \label{C_Omega_theorem}
    Under the same hypothesis of Theorem~\ref{Main_theorem}, $\mathcal{C}_\Omega\subset \mathcal{U}$. 
    $\hfill\square$
\end{lemma}
\begin{proof}
    For all $x \in \mathcal{C}_\Omega$, 
    \begin{align*}
        \begin{split}
W(x) & = (1 + \theta \sigma) V(x) - k\\
& \leq (1 + \theta \sigma _2)v_2 - (1+\theta \sigma_2)v_2 =0
\end{split}
    \end{align*}
    which completes the proof. 
\end{proof}

\section{Application: Safe Feedback Linearization for Constrained MIMO Systems}
As an important application,  in this section we illustrate how the proposed CLBF design method in the previous section can be employed to develop a nonlinear controller that resolves safe stabilization problem (possibly for nonlinear multi-input multi-output (MIMO) systems).
In particular, we will pay attention to developing a safety-guaranteeing version of the feedback linearization (FL)-based control, and to revealing that it can be applied to a safety-critical task space control for robot manipulators. 

\subsection{CLBF-based Safe Feedback Linearization}\label{subsec:safe_FL}
We consider a nonlinear MIMO system  having the form
\begin{subequations}\label{original_system}
    \begin{align} 
    \dot{\bf x}_1 & = {\bf x}_2,\\
    \dot{\bf x}_2 & = {\bf F}({\bf x}) + {\bf G}({\bf x}) {\bf u}
\end{align}
\end{subequations}
where ${\bf x}_i\in \mathbb{R}^n$ is the state, ${\bf u}\in \mathbb{R}^n$ is the input, and the square matrix ${\bf G}({\bf x}) \in \mathbb{R}^{n\times n}$ is invertible in a region of interest $\mathcal{X}$. 
Since \eqref{original_system} is feedback linearizable, a control law
\begin{align}\label{eq:FL}
    \phi({\bf x}):= {\bf G}^{-1}({\bf x})\big( -{\bf F}({\bf x}) - {\bf K}{\bf x} \big)
\end{align}
can easily stabilize the system \eqref{original_system}, where ${\bf K}$ is selected such that ${\bf A} - {\bf B}{\bf K}$ is Hurwitz, and
\begin{align*}
    {\bf A}:= \begin{bmatrix}
        {\bf 0} & {\bf I}\\
        {\bf 0} & {\bf 0} 
    \end{bmatrix}\in \mathbb{R}^{2n\times 2n},\quad {\bf B}:=\begin{bmatrix}
        {\bf 0}\\
        {\bf I}
    \end{bmatrix}\in \mathbb{R}^{2n\times n}.
\end{align*}

Now we assume that the system \eqref{original_system} has a set of linear inequality constraints
\begin{align}\label{eq:constraints}
    {\bf C}_{i} {\bf x} > d_i,\quad \forall i =1,\dots, m,
\end{align}
where ${\bf C}_i\in \mathbb{R}^{1\times n}$, $d_i\in \mathbb{R}$, and the number of the constraints, $m$, is assumed to be smaller than or equal to $n$: 
in other words, it is required that the controlled system is desired to be safe (in the sense of Definition~\ref{def:safety}) with respect to the unsafe set
\begin{align}\label{eq:unsafe_set_MIMO}
    \mathcal{D}:= \bigcup_{i=1}^{m}\{ {\bf x}\in \mathcal{X}: {\bf C}_i {\bf x} \leq d_i \} 
\end{align}
that is an union of hyperplanes. 
Without loss of generality, suppose that $\{{\bf C}_i, i=1,\dots,m\}$ is linearly independent. 
Obviously the existing control law ${\bf u} = \phi({\bf x})$ may not guarantee to satisfy the constraints \eqref{eq:constraints}, which allows us to refine the feedback linearization-based control as
\begin{align}\label{eq:controller_2}
    {\bf u}= \phi({\bf x}) + {\bf u}_{\rm safe}
\end{align}
where ${\bf u}_{\rm safe}\in \mathbb{R}^n$ represents an extra control input that will be utilized to achieve safety against the unsafe set $\mathcal{D}$. 

From now on, we will see that the proposed CLBF design in Section~\ref{sec:basic} is applicable to construction of such a ${\bf u}_{\rm safe}$. 
Since the theory in Section~\ref{sec:basic} is developed for scalar systems, it is first needed to introduce a coordinate change
\begin{align*}
    \overline{\bf x} :=\begin{bmatrix}
        \overline{\bf x}_1\\
        \overline{\bf x}_2
    \end{bmatrix} = \underbrace{\begin{bmatrix}
        {\bf P} & {\bf 0}\\
        {\bf 0} & {\bf P}
    \end{bmatrix}}_{=:{\bf \Phi}}\begin{bmatrix}
        {\bf x}_1\\
        {\bf x}_2
    \end{bmatrix},\quad {\bf P}:= \begin{bmatrix}
        {\bf C}\\
        {\bf D}
    \end{bmatrix}\in \mathbb{R}^{n}
\end{align*}
where the rows of ${\bf C}\in \mathbb{R}^{m\times n}$ are the very (linearly independent) row vectors ${\bf C}_i$, and ${\bf D}$ is selected such that ${\bf P}$ is nonsingular. 
In the new coordinate $\overline{\bf x}$, the system \eqref{original_system} is converted into
\begin{subequations}\label{eq:converted_system}
    \begin{align}
    \dot{\overline{\bf x}}_1 & = \overline{\bf x}_2,\\
    \dot{\overline{\bf x}}_2 & = {\bf P}{\bf F}( {\bf \Phi}^{-1} \overline{\bf x}) + {\bf P}{\bf G}({\bf \Phi}^{-1}\overline{\bf x}) {\bf u}
\end{align}
\end{subequations}
with the unsafe set
\begin{align}\label{eq:converted_unsafe_set}
    \overline{\mathcal{D}}:=   \bigcup_{i=1}^{m}\overline{\mathcal{D}}_i,\quad \overline{\mathcal{D}}_i:=\{ \overline{\bf x}\in \mathcal{X}: \overline{x}_{1i} \leq d_i \}.
\end{align}
Applying \eqref{eq:controller_2} and \eqref{eq:FL} to \eqref{eq:converted_system} then leads to
\begin{subequations}\label{eq:converted_system2}
    \begin{align}
    \dot{\overline{\bf x}}_1 & = \overline{\bf x}_2,\\
    \dot{\overline{\bf x}}_2 & = -{\bf P}{\bf K} {\bf \Phi}^{-1}\overline{\bf x} + {\bf P}{\bf G}( {\bf \Phi}^{-1}\overline{\bf x} ){\bf u}_{\rm safe}.
\end{align}
\end{subequations}
We now take ${\bf K}$ and ${\bf u}_{\rm safe}$ as 
\begin{align}
{\bf K} & = {\bf P}^{-1}{\rm blkdiag}\big( {\bf k}_1,\dots, {\bf k}_n \big){\bf \Phi}^{-1},\\
   {\bf u}_{\rm safe} & = {\bf G}^{-1}({\bf x}){\bf P}^{-1}{\bf a}_{\rm safe}\label{eq:u_safe}
\end{align}
with an auxiliary input ${\bf a}_{\rm safe}$, ${\bf k}_i = [k_{{\rm p},i}~k_{{\rm d},i}]$, $k_{{\rm p},i}>0$, and $k_{{\rm d},i}>0$, which brings $n$ {\it decoupled} subsystems
\begin{subequations}\label{eq:decoupled_system}
\begin{align}
    \dot{\overline{x}}_{1i} & = \overline{x}_{2i},\\
    \dot{\overline{x}}_{2i} & = -k_{{\rm p},i}\overline{x}_{1i} -k_{{\rm d},i}\overline{x}_{2i} + a_{{\rm safe},i}
\end{align}
\end{subequations}
and $m$ {\it decoupled} unsafe sets $\overline{\mathcal{D}}_i$ in \eqref{eq:converted_unsafe_set}. 
Notice that each subsystem \eqref{eq:decoupled_system} is exponentially stable by construction, so that the conditions in Assumption~\ref{asm:CLF} naturally hold. 

We are now ready to employ the findings in the previous section, by which one can obtain a set of the CLBFs, say $W_i$, $i=1,\dots,m$, with respect to each $\overline{\mathcal{D}}_i$. 
Note that Theorem~\ref{Main_theorem} ensures the existence of a nonlinear scaling-based CLBF $W_i$ in \eqref{CLBF_nonlinear_scaled} for the $i$-th decoupled subsystems \eqref{eq:decoupled_system}, with respect to $\overline{\mathcal{D}}_i$.
Then the safety-guaranteeing input ${\bf a}_{{\rm safe}}$ is chosen as
\begin{equation} \label{eq:sontag}
         a_{{\rm safe},i} = \begin{cases}
             k_{{\rm safe},i}\cdot\kappa_{\rm SU}\big(L_{\overline{F}_i}W_i, (L_{\overline{G}_i}W_i)^\top \big), & i=1,\dots,m\\
             0, & \text{otherwise}
         \end{cases}
\end{equation}
where $k_{{\rm safe},i}>0$ stands for the control gain for $a_{{\rm safe},i}$,
\begin{align*}
    \overline{F}_i(\overline{x}_i) := \begin{bmatrix}
        \overline{x}_{2i}\\
        k_{{\rm p},i}\overline{x}_{1i} + k_{{\rm d},i}\overline{x}_{2i}
    \end{bmatrix},\quad \overline{G}_i(\overline{x}_i) = \begin{bmatrix}
        0\\
        1
    \end{bmatrix}
\end{align*}
and $ \kappa_{\rm SU}(a,b)$ is the Sontag's universal formula 
\begin{align*}
    \kappa_{\rm SU}(a,b) := 
    \begin{cases}
     -\frac{a+\sqrt{a^2 + ||b||^4}}{b^{\top}b}b,  \quad& \text{if $b \neq 0$},\\
     0 & \text{otherwise}.
    \end{cases}
\end{align*}
It is well-studied in the literature (e.g., \cite{romdlony2016stabilization}) that for given CLBF $W_i$, the Sontag's universal formula-based controller \eqref{eq:sontag} guarantees that ${W}_i$ decreases as time goes on, which means that the safe stabilization is achieved. 
From this perspective, we call the modified FL-based controller \eqref{eq:controller_2}, \eqref{eq:u_safe}, and \eqref{eq:sontag} a {\it safe FL-based controller}. 

The following corollary summarizes the discussions that we have made so far. 
\begin{corollary} 
    Suppose that $\mathcal{D}$ in \eqref{eq:unsafe_set_MIMO} does not contain the origin ${\bf 0}$. 
    Then for the unsafe set $\mathcal{D}$ in \eqref{eq:unsafe_set_MIMO}, the state ${\bf x}(t)$ of the system \eqref{original_system} controlled by a safe FL-based controller \eqref{eq:controller_2}, \eqref{eq:u_safe}, and \eqref{eq:sontag} initiated in 
    \begin{align*}
        \{ {\bf x} = {\bf \Phi}^{-1}\overline{\bf x}:  W_i(\overline{x}_{i})\leq 0, \forall i=1,\dots,m \}
    \end{align*}satisfies
    \begin{itemize}
        \item (Stability) ${\bf x}(t)\rightarrow {\bf 0}$ as $t\rightarrow \infty$, and;
        \item (Safety) ${\bf x}(t) \notin {\rm cl}(\mathcal{D})$ for all $t\geq 0$.    $\hfill \square$
    \end{itemize}
\end{corollary}

\begin{remark}
    It is remarkable that the expression \eqref{eq:controller_2} introduces a possibility of utilizing a CLBF-based safety-critical control law ${\bf u}_{\rm safe}$ as an {\it add-on type control}: in other words, one can ensure the safety of a controlled system simply by adding an extra input ${\bf u}_{\rm safe}$ into the existing control loop with $\phi(x)$. 
    From this perspective, the generating law for ${\bf u}_{\rm safe}$ in \eqref{eq:controller_2} can be interpreted as a new type of the {\it safety filter} that has been extensively studied in the literature. 
    $\hfill\square$
\end{remark}



\subsection{Example: Safe Task-space Control for Two-link Robot Manipulator}

The proposed safe FL is applicable to a wide range of industrial problems where the system to be controlled has the form \eqref{original_system}.  
As an example, in this subsection we address the problem of task-space control for planar two-link robot manipulator, particularly with guaranteed safety. 
To this end, consider a robot dynamics
 \begin{equation}  \label{manipulator_dynamics_joint}
      {\bf M}({\bf q})\ddot{{\bf q}} + {\bf c}({\bf q},\dot {\bf q}) + {\bf g}({\bf q}) = {\bm \tau}
 \end{equation}
with the joint angle ${\bf q} := (\theta_1, \theta_2) \in \mathbb{R}^2$, the control input ${\bm \tau}\in \mathbb{R}^2$, the mass matrix ${\bf M}$ given by
\begin{displaymath}
    {\bf M}({\bf q}) = 
    \begin{bmatrix}
        {\bf M}_{11} & {\bf M}_{12}\\
        {\bf M}_{21} & {\bf M}_{22}
    \end{bmatrix} = {\bf M}({\bf q})^\top >0
\end{displaymath}
{
where
\begin{align*}
    & {\bf M}_{11} := m_1 L_1^2 + m_2 (L_1^2 + 2L_1 L_2 \cos\theta_2 + L_2^2)\\
    & {\bf M}_{12} := {\bf M}_{21} = m_2 (L_1 L_2 \cos\theta_2 + L_2^2), \: {\bf M}_{22} := m_2 L_2^2,
\end{align*}
}
the torque vector related to the centrifugal and Coriolis effects
\begin{displaymath}
{\bf c}({\bf q}, \dot{{\bf q}}) =
\begin{bmatrix}
- m_2 L_1 L_2 \sin\theta_2 (2\dot{\theta}_1 \dot{\theta}_2 + \dot{\theta}_2^2) \\
m_2 L_1 L_2 \dot{\theta}_1^2 \sin\theta_2
\end{bmatrix},
\end{displaymath}
and the gravitational torque
\begin{displaymath}
{\bf g}({\bf q}) =
\begin{bmatrix}
(m_1 + m_2) L_1 {\rm g} \cos\theta_1 + m_2 {\rm g} L_2 \cos(\theta_1 + \theta_2) \\
m_2 {\rm g} L_2 \cos(\theta_1 + \theta_2)
\end{bmatrix}.
\end{displaymath}
In the above equations, $m_1$ and $m_2$ represent the masses of each body,  $L_1$ and $L_2$ are the length of the bodies, and $\textrm{g}$ is the gravitational acceleration. 

The forward kinematics of the robot maps the joint angle  ${\bf q}\in \mathbb{R}^2$ to the Cartesian position ${\bf p} :=({\bf p}_1, {\bf p}_2)\in \mathbb{R}^2$ of the end-effector of the robot  as follows:
\begin{equation*}
    {\bf p} := \psi({\bf q}) = \begin{bmatrix}
L_1\cos\theta_1 + L_2\cos(\theta_1 + \theta_2)
\\
L_1\sin\theta_1 + L_2\sin(\theta_1 + \theta_2)
\end{bmatrix}.
\end{equation*}
Then one readily derives the relation between the joint velocity $\dot{\bf q}$ and the Cartesian velocity $\mathbf v = \dot {\mathbf{p}}$ as $ {\mathbf v} = {\mathbf J} \dot {\mathbf q}$ with the Jacobian matrix ${\bf J} :=(\partial {\psi})/(\partial {\bf q})$ given by  
\begin{displaymath}
{\bf J} =  \begin{bmatrix}
-L_1\sin\theta_1 - L_2\sin(\theta_1 + \theta_2)
& -L_2\sin(\theta_1 + \theta_2)
\\
L_1\cos{\theta}_1 + L_2\cos(\theta_1 + \theta_2)
&  L_2\cos(\theta_1 + \theta_2)
\end{bmatrix}.
\end{displaymath}
After some computations, we express the robot dynamics \eqref{manipulator_dynamics_joint} written in the joint space into the task-space dynamics in the Cartesian coordinate (at least on the region where ${\bf J}$ has full row rank)
 \begin{equation} \label{2link_dynamics_cartesian}
      {\bf M}_{{\bf p}}({\bf q})\ddot{{\bf p}} + {\bf c}_{{\bf p}}({\bf q},\dot {\bf q}) + {\bf g}_{{\bf p}}({\bf q}) = \mathcal{F},
 \end{equation}
where $\mathcal{F} := {\bf J}^{-\top}{\bm \tau} \in \mathbb{R}^2$ is a force acting on the end-effector (which can be regarded as a control input), and
\begin{align*} 
{\bf M}_{{\bf p}} & := {\bf J}^{-\top}{\bf M}{\bf J}^{-1}, ~ {\bf c}_{{\bf p}}:= -{\bf M}_{{\bf p}}\dot{{\bf J}}\dot {\bf q} + {\bf J}^{-\top}{\bf c},\\
{\bf g}_{{\bf p}} & := {\bf J}^{-\top} {\bf g}.
\end{align*}

For the robot manipulator whose dynamics is expressed in the Cartesian space as in \eqref{2link_dynamics_cartesian}, we aim to finding a safe FL-based controller $\mathcal{F}$ that solves the regulation problem defined in the Cartesian space for a desired point ${\bf p}_{\rm d}$,
while guaranteeing the constraints
\begin{align}\label{example_constraints}
    {\bf p}_1 < -\hat{d}_1,\quad {\bf p}_2 > \hat{d}_2
\end{align}
for some $\hat{d}_i\in \mathbb{R}$. 
To apply the safe FL theory in the previous subsection to this example, define
\begin{align*}
    {\bf x} = \begin{bmatrix}
        {\bf x}_1\\
        {\bf x}_2
    \end{bmatrix},\quad \text{where}~~{\bf x}_1 = \begin{bmatrix}
        -{\bf p}_1 + {\bf p}_{{\rm d},1}\\
        {\bf p}_2 - {\bf p}_{{\rm d},2}
    \end{bmatrix},~~
    {\bf x}_2 = \begin{bmatrix}
        -\dot{\bf p}_1\\
        \dot{\bf p}_2
    \end{bmatrix},
\end{align*}
which admits the unsafe set ${\mathcal{D}}$ to have the form $\overline{\mathcal{D}}$ in \eqref{eq:converted_unsafe_set} with $d_i = {\bf p}_{{\rm d},i} - \hat{d}_i$, $i=1,2$. 
Moreover, the task-space dynamics \eqref{2link_dynamics_cartesian} can be represented as \eqref{original_system} (or equivalently, \eqref{eq:converted_system}) with ${\bf \Phi}=I$, while the rest expansion is omitted due to clarity of explanation.


Finally, by following the same procedure of Subsection~\ref{subsec:safe_FL}, we obtain a safe FL-based controller
\begin{subequations} \label{FL_force}
\begin{align} 
\mathcal{F} & = \phi({\bf x}) + \mathcal{F}_{\rm safe}, \label{FL_force_a}\\
\phi({\bf x}) & = {\bf M}_{{\bf p}}({\bf q}) \left(  - {\bf K}{\bf x}\right) + {\bf c}_{{\bf p}}({\bf q},\dot {\bf q}) + {\bf g}_{{\bf p}}({\bf q}),\label{FL_force_b}\\
\mathcal{F}_{\rm safe} & = {\bf M}_{\bf p}({\bf q}) {\bf a}_{\rm safe} \label{FL_force_c}
\end{align}
\end{subequations}
where ${\bf a}_{\rm safe}$ is designed as \eqref{eq:sontag}, and each $W_i$ has the form \eqref{CLBF_nonlinear_scaled} derived from the nonlinear scaling-based design method. 

\subsection{Simulation Result}
Now we perform a simulation for numerical case of given example to verify the validity of proposed control scheme. Consider system \eqref{manipulator_dynamics_joint} with $m_1 = m_2 = 0.8$ and $L_1 = L_2 = 1$ which is assumed to be operated in a given compact set $\mathcal{X} := \{(\mathbf{p}, \mathbf v) \in \mathbb{R}^4 : -0.2 \leq {\bf p}_{1} \leq 1.5, -1.0 \leq {\bf p}_{2} \leq 0.5, -2.5\leq \mathbf{v}_i \leq 2.5 \quad (i=1,2)\} $. With initial condition $\mathbf{p} = (1.0, 0.4)$ and $\mathbf{v} = (1.5, -2,5)$, the system is expected to maintain desired state $\mathbf{p}_{\mathrm{d}} = (0.3, 1.0)$ while guaranteeing the constraints \eqref{example_constraints} with $\hat d_1= -1.3, \hat d_2 = -0.3$.  

To construct the safe FL-based controller, first take ${\bf \Phi} = I_2$ and ${\bf K}$ with $k_{\mathrm{p},i}$, $k_{\mathrm{d},i}$ shown in Table I for \eqref{FL_force_b} to decouple a given system. Then CLFs of decoupled subsystems \eqref{eq:decoupled_system} in form of \eqref{asm:stability_a} are able to be found by employing  $P_i$ for each $P$ of CLFs, which satisfies
\begin{align*}
 \begin{bmatrix}
     0 & -k_{\mathrm{p},i}\\
     1 & -k_{\mathrm{d},i}
 \end{bmatrix}P_i + P_i\begin{bmatrix}
     0 & 1\\
     -k_{\mathrm{p},i} & -k_{\mathrm{d},i}
 \end{bmatrix}    = 
   -  \begin{bmatrix}
         1.0 & -0.9\\
         -0.9 & 1.0
     \end{bmatrix},
\end{align*}
$i =1,2$.
Next, using the CLFs found, select $l_i, \delta_i, \theta_i, k_i$ in Table I to construct nonlinear-scaled CLBFs. At the initial point, the CLBF values are $W_1=-0.43, W_2 = -2.33$, confirming that initial condition satisfies the set condition in Corollary 1. Apply \eqref{eq:sontag} for safety input \eqref{FL_force_c}.

\begin{table}[]
    \centering
    \begin{tabular}{|c||c|c|c|c|c|c|}
    \hline
          & $k_{{\rm p},i}$& $k_{{\rm d},i}$& $l_i$ & $\delta_i$ & $\theta_i$ & $k_i$ \\
    \hline
    \hline
         $i=1$ & 1.5 & 1.0 & 4.0 & 0.28 & 50.0 & 12.7\\
    \cline{1-7}
         $i=2$ & 1.0 & 1.0 & 4.0 & 0.58 & 6.1 & 10.67 \\
    \hline
    \end{tabular}
    \caption{Control parameters used in simulation}
    \label{tab:my_label}
\end{table}
\begin{figure}[h!]
    \centering
    \includegraphics[width=3.0in]{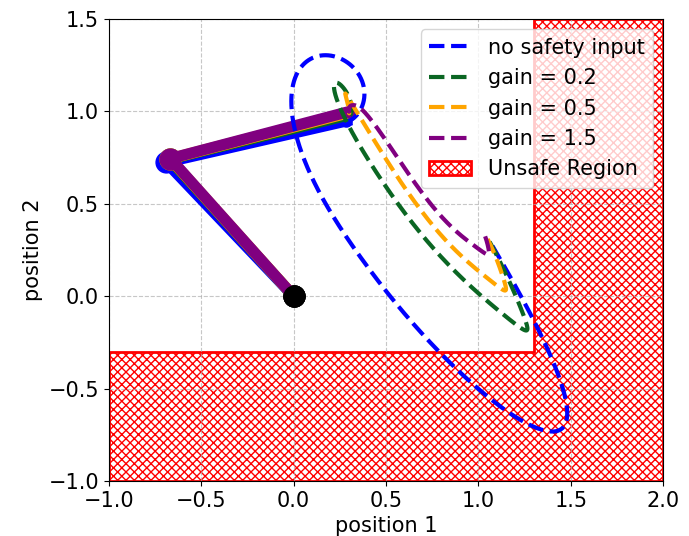}
    \caption{
    {
    End-effector trajectories of two-link manipulators controlled by the conventional (non-safe) FL-based controller (blue) and the proposed safe FL-based controllers with various values of $\mathbf{k}_{\textrm{safe},i}=0.2$, $0.5$, and $1.5$ (applied equally for $i=1,2$): 
    Position 1 and 2 in the figure represents $\mathbf{p}_1$ and  $\mathbf{p}_2$, respectively. 
    Increasing the safety gain drives the position of the end-effector to avoid the unsafe set (red shaded).
    }
    }
    \label{fig:trajectory}
\end{figure}
\begin{figure}
\centering
    \subfigure[All operation duration]{\includegraphics[width=0.9\linewidth]{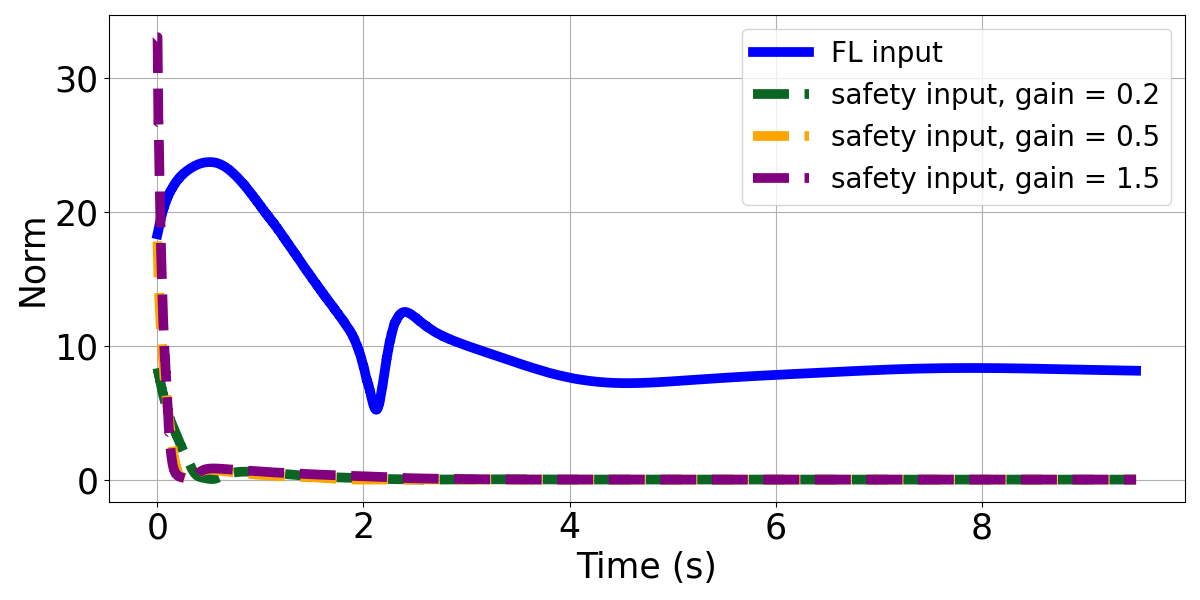}}
    \hfill
    \subfigure[Near initial condition]{\includegraphics[width=0.9\linewidth]{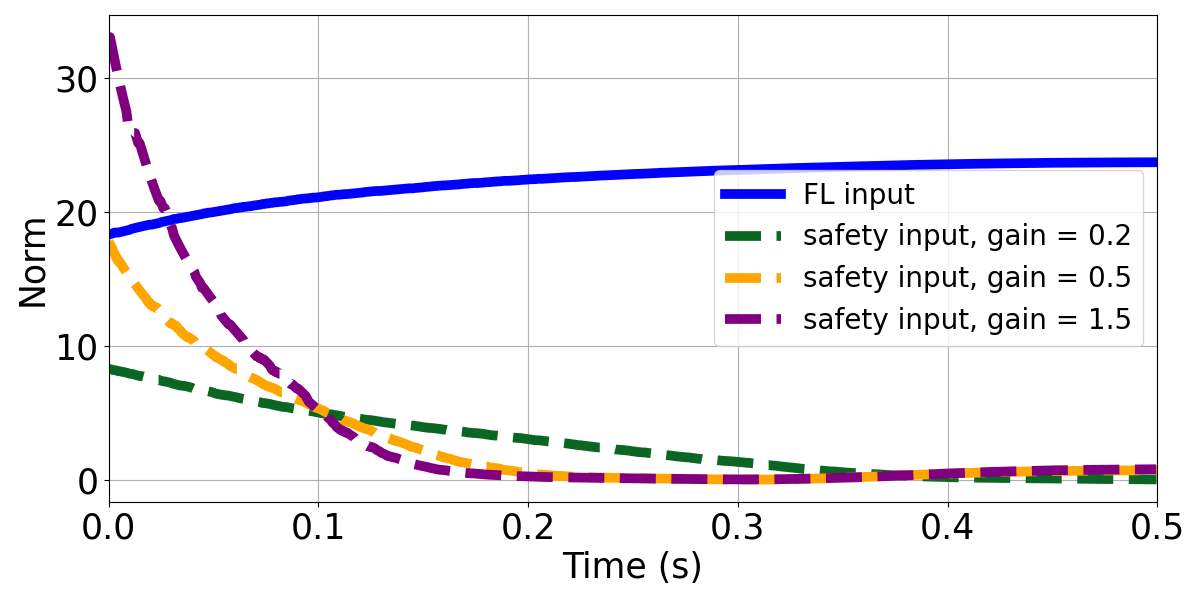}}
    \caption{Input norm of the only applied feedback linearization, and norm of the added $\mathcal{F}_{\textrm{safe}}$ for each case of different gains $k_{\textrm{safe},i}$. The extra safety input $\mathcal{F}_{\textrm{safe}}$ has significant effect only while avoiding to enter the unsafe sets, while its influence becomes minor when the state is far from the unsafe sets.}
    \label{fig:Input_2norm}
\end{figure}





Fig.~\ref{fig:trajectory} shows the resulting trajectory of safe FL-base control with various safety gains $\mathbf{k}_{\mathrm{safe},i}$, compared to the trajectory when only FL control $\phi(x)$ is applied.
Fig.~\ref{fig:Input_2norm} plots comparison of the input force norms between $\phi(x)$ and $\mathcal{F}_{\textrm{safe}}$ corresponding to the situations in Fig.~\ref{fig:trajectory}.

\section{Conclusion}
In this paper, we proposed a systematic design method of a CLBF for $2$nd-order systems having the relative degree 2, via the nonlinear scaling based on a sigmoid function. 
A basic result was first developed for single-input cases, while its extension to a class of multi-input cases was also discussed. 
It was in turn seen that the theory for the CLBF design can be applied to refining the existing FL method to be safe with respect to an unsafe set that can be expressed as an union of  hyperplanes. 
We showed that this safe FL approach allows to address the safe stabilization problem for a robot manipulator in the task space. 
Future works will include extension of the present result to higher-order cases as well application to safe control of robots in real world.


\addtolength{\textheight}{-12cm}   



\section*{Appendix}

We claim that all the $(i,j)$-component $p_{ij}=[P]_{ij}$ of the square matrix $P$ satisfying Assumption~\ref{asm:CLF} are positive. 
Since it is obvious that $p_{11}, p_{22} >0$ by the positive definiteness of $P$, 
it is enough to see $p_{12}=p_{21}>0$.
Recall that the set of $x$ satisfying $L_GV(x) =0$ forms a subspace where \eqref{L_GW=0_condition} holds. 
On the subspace, $L_FV(x)$ (that is expected to be negative) is computed as
\begin{align*}
L_FV(x)&= \dfrac{\partial V}{\partial x_1}x_2 + \dfrac{\partial V}{\partial x_2} f(x) = \dfrac{\partial V}{\partial x_1} x_2\\
& =(p_{11}x_1 + p_{12}x_2)x_2 = p_{11}x_1x_2 + p_{12}x_2^2 \\
& = -p_{12}x_1^2 \Big(\dfrac{p_{11}}{p_{22}} - \dfrac{p_{12}^2}{p_{22}^2} \Big) = -\dfrac{p_{12}}{p_{22}^2}x_1^2\left(
{\rm det}(P) 
\right)
\end{align*}
Noting that $L_F V(x)<0$ by the CLF condition and  ${\rm det}(P) = p_{11}p_{12} - p_{12}^2>0$, it follows that $p_{12} >0$.


\nocite{*}
\bibliographystyle{IEEEtran}
\bibliography{IEEEabrv, references}

\begin{thebibliography}{10}
\providecommand{\url}[1]{#1}
\csname url@samestyle\endcsname
\providecommand{\newblock}{\relax}
\providecommand{\bibinfo}[2]{#2}
\providecommand{\BIBentrySTDinterwordspacing}{\spaceskip=0pt\relax}
\providecommand{\BIBentryALTinterwordstretchfactor}{4}
\providecommand{\BIBentryALTinterwordspacing}{\spaceskip=\fontdimen2\font plus
\BIBentryALTinterwordstretchfactor\fontdimen3\font minus
  \fontdimen4\font\relax}
\providecommand{\BIBforeignlanguage}[2]{{%
\expandafter\ifx\csname l@#1\endcsname\relax
\typeout{** WARNING: IEEEtran.bst: No hyphenation pattern has been}%
\typeout{** loaded for the language `#1'. Using the pattern for}%
\typeout{** the default language instead.}%
\else
\language=\csname l@#1\endcsname
\fi
#2}}
\providecommand{\BIBdecl}{\relax}
\BIBdecl

\bibitem{wieland2007constructive}
P.~Wieland and F.~Allg{\"o}wer, ``Constructive safety using control barrier
  functions,'' \emph{IFAC Proceedings Volumes}, vol.~40, no.~12, pp. 462--467,
  2007.

\bibitem{ames2019control}
A.~D. Ames, S.~Coogan, M.~Egerstedt, G.~Notomista, K.~Sreenath, and P.~Tabuada,
  ``Control barrier functions: Theory and applications,'' in \emph{2019 18th
  European control conference (ECC)}.\hskip 1em plus 0.5em minus 0.4em\relax
  IEEE, 2019, pp. 3420--3431.

\bibitem{romdlony2016stabilization}
M.~Z. Romdlony and B.~Jayawardhana, ``Stabilization with guaranteed safety
  using control lyapunov--barrier function,'' \emph{Automatica}, vol.~66, pp.
  39--47, 2016.

\bibitem{romdlony2014uniting}
------, ``Uniting control lyapunov and control barrier functions,'' in
  \emph{53rd IEEE Conference on Decision and Control}.\hskip 1em plus 0.5em
  minus 0.4em\relax IEEE, 2014, pp. 2293--2298.

\bibitem{wu2019control}
Z.~Wu, F.~Albalawi, Z.~Zhang, J.~Zhang, H.~Durand, and P.~D. Christofides,
  ``Control lyapunov-barrier function-based model predictive control of
  nonlinear systems,'' \emph{Automatica}, vol. 109, p. 108508, 2019.

\bibitem{lin1991universal}
Y.~Lin and E.~D. Sontag, ``A universal formula for stabilization with bounded
  controls,'' \emph{Systems \& control letters}, vol.~16, no.~6, pp. 393--397,
  1991.

\bibitem{xiao2021high}
W.~Xiao, C.~A. Belta, and C.~G. Cassandras, ``High order control
  lyapunov-barrier functions for temporal logic specifications,'' in \emph{2021
  American Control Conference (ACC)}.\hskip 1em plus 0.5em minus 0.4em\relax
  IEEE, 2021, pp. 4886--4891.

\bibitem{braun40899existence}
P.~Braun and C.~Kellett, ``On (the existence of) control lyapunov barrier
  functions,'' \emph{Preprint}, vol. 40899, 2017.

\bibitem{craig2005introduction}
J.~Craig, \emph{Introduction to Robotics: Mechanics and Control}.\hskip 1em
  plus 0.5em minus 0.4em\relax Pearson/Prentice Hall, 2005.

\bibitem{khalil2015nonlinear}
H.~Khalil, \emph{Nonlinear Control}.\hskip 1em plus 0.5em minus 0.4em\relax
  Pearson, 2015.

\bibitem{xu2015robustness}
X.~Xu, P.~Tabuada, J.~W. Grizzle, and A.~D. Ames, ``Robustness of control
  barrier functions for safety critical control,'' \emph{IFAC-PapersOnLine},
  vol.~48, no.~27, pp. 54--61, 2015.

\bibitem{he2021rule}
S.~He, J.~Zeng, B.~Zhang, and K.~Sreenath, ``Rule-based safety-critical control
  design using control barrier functions with application to autonomous lane
  change,'' in \emph{2021 American Control Conference (ACC)}.\hskip 1em plus
  0.5em minus 0.4em\relax IEEE, 2021, pp. 178--185.

\end{thebibliography}


\end{document}